\documentclass[conference,a4paper]{IEEEtran}  

\usepackage[utf8]{inputenc} 
\usepackage[T1]{fontenc}
\usepackage{url}
\usepackage{ifthen}
\usepackage{cite}
\usepackage[cmex10]{amsmath} 
\usepackage{amssymb,mathtools}
\usepackage{graphicx}
\usepackage{multirow}
\usepackage{amsthm}
\usepackage{cuted}
\usepackage{amssymb, latexsym, bm, mathtools, braket, multirow,enumitem}

\usepackage{subcaption}

\usepackage{color}

\DeclareFontFamily{OMX}{MnSymbolE}{}
\DeclareFontShape{OMX}{MnSymbolE}{m}{n}{
    <-6>  MnSymbolE5
   <6-7>  MnSymbolE6
   <7-8>  MnSymbolE7
   <8-9>  MnSymbolE8
   <9-10> MnSymbolE9
  <10-12> MnSymbolE10
  <12->   MnSymbolE12}{}
\DeclareSymbolFont{mnlargesymbols}{OMX}{MnSymbolE}{m}{n}
\SetSymbolFont{mnlargesymbols}{bold}{OMX}{MnSymbolE}{b}{n}
\DeclareMathDelimiter{\llangle}{\mathopen}{mnlargesymbols}{'164}{mnlargesymbols}{'164}
\DeclareMathDelimiter{\rrangle}{\mathclose}{mnlargesymbols}{'171}{mnlargesymbols}{'171}

\theoremstyle{plain}
\newtheorem{lemm}{Lemma}

\newtheorem{theo}{Theorem}
\newtheorem{prop}{Proposition}

\newtheorem{coro}{Corollary}

\theoremstyle{definition}
\newtheorem{remark}{Remark}
\newtheorem*{remark*}{Remark}
\newtheorem{prot}{Protocol}

\newcommand{\cD}{\mathcal{D}}
\newcommand{\cE}{\mathcal{E}}

\newcommand{\cP}{\mathcal{P}}
\newcommand{\cQ}{\mathcal{Q}}

\newcommand*{\QEDA}{\hfill\ensuremath{\blacksquare}}%
\newcommand{\ii}{\iota}

\newcommand{\pureset}[1]{\mathrm{P}(#1)}

\newcommand{\U}[1]{\mathrm{U}(#1)}
\newcommand{\Rr}{\sR}
\newcommand{\Ss}{\sS}

\DeclareMathOperator{\Tr}{Tr}

\DeclareMathOperator{\id}{id}

\DeclarePairedDelimiter\px{\{}{\}}
\DeclarePairedDelimiter\paren{(}{)}

\newcommand\vN{\mathsf{n}}
\newcommand\vF{\mathsf{f}}

\newcommand\vM{\mathsf{m}}
\newcommand\vR{\mathsf{r}}

\newcommand\CC{\mathrm{CC}(\Phi)}

\newcommand\sA{\mathsf{A}}
\newcommand\sB{\mathsf{B}}
\newcommand\sC{\mathsf{C}}

\newcommand\sI{\mathsf{I}}
\newcommand\sM{\mathsf{M}}
\newcommand\sR{\mathsf{R}}
\newcommand\sS{\mathsf{S}}
\newcommand\sU{\mathsf{U}}

\newcommand\sX{\mathsf{X}}
\newcommand\sY{\mathsf{Y}}
\newcommand\sZ{\mathsf{Z}}

\usepackage{tikz}
\usetikzlibrary{arrows}
\usetikzlibrary{shapes,snakes, quantikz}
\usetikzlibrary{tikzmark}
\usetikzlibrary{fit,backgrounds}
\usetikzlibrary{matrix,decorations.pathreplacing,angles,quotes}
\usetikzlibrary{patterns}
\usetikzlibrary{calc}
\usetikzlibrary{positioning}
\tikzstyle{block} = [rectangle, draw, 
    minimum width=3em, text centered, minimum height=12em]
\tikzstyle{rblock} = [rectangle, draw, 
    minimum width=5em, text centered, minimum height=4em, rounded corners]   
\tikzstyle{circ} = [circle, draw, inner sep=0em, minimum width=2em,
     text centered]
\tikzstyle{circ2} = [circle, draw, inner sep=0em, minimum width=2.2em,
     text centered]    
\tikzstyle{rec} = [rectangle, draw]
\tikzstyle{line} = [draw, -latex]
\tikzstyle{imp-line} = [draw, -implies, double equal sign distance]

\tikzset{snake arrow/.style=
{
decorate,
decoration={snake,amplitude=.4mm,segment length=2mm,pre length=1mm, post length=1mm}}
}
\tikzset{snake arrow0/.style=
{
decorate,
decoration={snake,amplitude=.4mm,segment length=2mm,post length=1mm}}
}

\begin{document}

\title{Quantum Private Information Retrieval for Quantum Messages}

\author{%
 \IEEEauthorblockN{Seunghoan Song$^{a}$
 and Masahito Hayashi$^{b,a}$}
  \IEEEauthorblockA{$~^{a}$Graduate School of Mathematics, Nagoya University \\
$^b$Shenzhen Institute for Quantum Science and Engineering, Southern University of Science and Technology\\
    Email: {m17021a@math.nagoya-u.ac.jp \& hayashi@sustech.edu.cn} }
}

\maketitle

\begin{abstract}
Quantum private information retrieval (QPIR) for quantum messages
	is the protocol in which a user retrieves one of the multiple quantum states from one or multiple servers without revealing which state is retrieved.
We consider QPIR in two different settings:
	the blind setting, in which the servers contain one copy of the message states,
	and the visible setting, in which the servers contain the description of the message states.
One trivial solution in both settings is downloading all states from the servers and the main goal of this paper is to find more efficient QPIR protocols.
First, we prove that the trivial solution is optimal for one-server QPIR in the blind setting.
	In one-round protocols, the same optimality holds even in the visible setting.
On the other hand, when the user and the server share entanglement,
	we prove that there exists an efficient one-server QPIR protocol in the blind setting.
{Furthermore, 
	in the visible setting, 
	we prove that it is possible to construct {\em symmetric} QPIR protocols in which the user obtains no information of the non-targeted messages.
We construct three two-server symmetric QPIR protocols for pure states.
Note that symmetric classical PIR is impossible without shared randomness unknown to the user.}
\end{abstract}


\section{Introduction}

\subsection{Backgrounds: Private information retrieval (PIR) for classical messages}

\subsubsection{Private information retrieval (PIR)}

Private information retrieval is a method to retrieve a message from a server without revealing which message is retrieved.
PIR has a simple solution of downloading all messages
	and this trivial solution is proved to be optimal \cite{CGKS98}, i.e.,
	the optimal communication complexity is $O(\vM)$, where $\vM$ is the total number of bits in the messages.
%
To improve the communication efficiency of PIR, 
	there have been mainly two approaches: PIR with computational assumptions \cite{CMS99, Lipmaa10} and PIR with multiple servers \cite{BS03,Yekanin07,DGH12}.
Recently, information-theoretic aspects of PIR has been extensively studied \cite
{CHY15,SJ17, SJ17-2, SJ18, BU18, FHGHK17, KLRG17, LKRG18, TSC18-2, WS17,  Tandon17,  BU19,L19, KGHERS19, TGKFH19}.
In this paper, we only focus on the one-server and multi-server PIRs without computational assumptions.

Symmetric PIR (SPIR) is PIR in which the user obtains only the targeted message but no information of the non-targeted messages.
In the one-server case, SPIR is also called oblivious transfer, which is proven to be impossible even in the quantum case \cite{Lo97}. 
In the multi-server case, Gertner et al. \cite{GIKM00} proved that SPIR is impossible, but it becomes possible if the servers share randomness.

\subsubsection{One-server quantum PIR for classical messages}

PIR has also been studied when quantum communication is allowed between the user and the server(s) \cite{KdW03,KdW04, Ole11,BB15, LeG12,KLLGR16, ABCGLS19, SH19, SH19-2, SH20, AHPH20,KL20}.
Hereinafter, we denote quantum PIR for classical messages as C-QPIR. 
Interestingly, when the server is honest, i.e., the server does not deviate from the protocol,
	Le Gall \cite{LeG12} proposed a C-QPIR protocol with communication complexity $O(\sqrt{\vM})$, and Kerenidis et al. \cite{KLLGR16} improved this result by $O(\mathrm{poly}\log\vM)$, 
	where the communication in the quantum case is the total number of communicated qubits.
However, when the server deviates from the protocol as far as its malicious operations are not revealed to the user, which is called {\em specious adversary},
	Baumeler and Broadbent \cite{BB15} proved that the communication complexity is at least $O(\vM)$, i.e., the trivial solution of downloading all messages is optimal also for this case.
When prior entanglement is allowed between the user and the server,
	the communication complexity is improved by $O(\log \vM)$
		for the honest server model \cite{KLLGR16}, 
	but the communication complexity is also lower bounded by $O(\vM)$ for the specious attack \cite{ABCGLS19}.

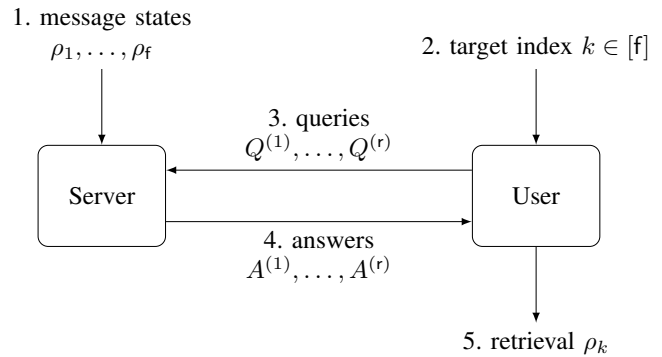
\begin{figure}[t]
\begin{center}
        \resizebox {1\linewidth} {!} {
\begin{tikzpicture}[scale=0.5, node distance = 3.3cm, every text node part/.style={align=center}, auto]
	\node [rblock] (serv) {Server};
	\node [rblock, right=12em of serv] (user) {User};
	\node [above=3em of serv] (1) {1. message states\\$\rho_1, \ldots , \rho_{\vF}$};
	\node [above=3em of user] (2) {2. target index $k\in[\vF]$};
	\node [below=3em of user] (5) {5. retrieval $\rho_k$};

	\path [line] (serv.22 -| user.west) --node[above]{3. queries\\$Q^{(1)},\ldots,Q^{(\vR)} $} (serv.22);
	\path [line] (serv.-22) --node[below]{4. answers\\$A^{(1)}, \ldots,A^{(\vR)}$} (serv.-22 -| user.west);
	\path [line] (1) -- (serv.north);
	\path [line] (2) -- (user.north);
	\path [line] (user.south) -| (5);
	
\end{tikzpicture}
}
\caption{One-server QPIR protocol with quantum messages. At round $i$, the user uploads a query $Q^{(i)}$ and downloads an answer $A^{(i)}$.}
\label{fig:one-server}
\end{center}
\end{figure}

\subsubsection{Multi-server quantum PIR for classical messages}


For multi-server C-QPIR, 
	Kerenidis and de Wolf \cite{KdW04}
		proposed a symmetric C-QPIR protocol without shared randomness.
Song and Hayashi \cite{SH19}
	derived symmetric and non-symmetric C-QPIR capacity,
	which is defined similar to its classical counterpart \cite{SJ17,SJ17-2} as the optimal communication efficiency for arbitrary-long classical messages,
	and proved that symmetric C-QPIR can be constructed without any communication loss if prior-entanglement among servers is allowed.
	They \cite{SH19-2,SH20} 
		and Allaix et al. \cite{AHPH20} 
		also considered C-QPIR with colluding servers in which 
		secrecy of the protocol is preserved even if some servers may communicate and collude.
Kon and Lim \cite{KL20} constructed a symmetric C-QPIR protocol with quantum-key distribution.


\begin{table}[t]
\begin{center}
\caption{Definition of symbols} \label{tab:symbols}
\begin{tabular}{|c|c|c|c|c|}
\hline
Symbol &	Definition	\\
\hline
$\vM$	&	Total size of messages (states)	\\
\hline
$\vF$	&	Number of messages (states)	\\
\hline
$\vR$	&	Number of rounds in multi-round models	\\
\hline
$\vN$	&	Number of servers in multi-server models	\\
\hline
\end{tabular}
\end{center}
\end{table}

\subsection{Contribution: quantum PIR for quantum messages}

Even if quantum PIR has been studied for classical messages,
	there has been no study of PIR for quantum messages, i.e., quantum states.
This paper considers quantum PIR for quantum messages.
Throughout this paper, we denote quantum PIR for quantum messages as QPIR.
Downloading all quantum messages is also a trivial solution for QPIR. 
Intuitively, it seems reasonable to conjecture this trivial solution is optimal for QPIR. 
However,
	since the C-QPIR for the honest server has more efficient solutions than the trivial solution \cite{LeG12,KLLGR16},
	we cannot exclude the possibility of efficient one-server QPIR protocols.
Furthermore, 
	whereas the optimality proof of classical PIR \cite{CGKS98} uses the communication transcript between the server and the user, 
	we cannot apply the same technique because quantum states cannot be copied because of the no-cloning theorem.
Thus, the first goal is to answer 
	whether the trivial solution is also optimal QPIR under several attack models
	and how we can construct more efficient protocols than the trivial solution.
Furthermore, it is not unknown if symmetric QPIR (QSPIR) is possible on the multiple server model. 

%

\subsubsection{Two settings in QPIR}
QPIR can be studied in two distinct settings, called the {\em blind} and {\em visible} settings, 
	in which quantum state compression has also been extensively studied \cite{423,288,290,269,35,201}.
In the blind setting, 
	the server(s) contains 
		quantum systems $X_1,\ldots,X_{\vF}$ with the message states $\rho_1,\ldots,\rho_{\vF}$, respectively,
	but does not know the states of the systems.
Due to the no-cloning theorem, the server(s) cannot generate more copies of the message states and the server's operations are independent of the message states.
QPIR in the blind setting is suitable for the case where 
	the server(s) generates the message states by some quantum algorithm and performs the QPIR task.



On the other hand, in the visible setting, 
	the server(s) contains the descriptions of the message states.
With the descriptions of quantum states, the servers can generate multiple copies of the quantum states, without the limitation of the no-cloning theorem, 
	and apply quantum operations depending on the descriptions of the states.
Since any protocol in the blind setting can be considered as a protocol in the visible setting, 
	we can generally expect to achieve lower communication complexity in the visible setting.
Even if the description of a quantum state is classical information, it is infinite-length classical information,
	and therefore we cannot send this description via QPIR for classical messages. 
Furthermore, the visible setting is a reasonable setting for the case where the user has no ability to generate quantum states
	and requires the generation of the targeted state along with the QPIR task.

\subsubsection{Optimality of trivial solution for one-server QPIR}
Our first result is that 
	the trivial solution is optimal for one-server QPIR for the honest server model.
For the one-server case, 
	the comparison of the our results and previous results are summarized in Table~\ref{tab:existing_results}.
We first prove for the one-round case
	and 
	then extend the result to the multi-round case.
In the one-round case, we prove the optimality both for blind and visible settings, but in the multi-round case, we only prove for the blind setting.
The entropic inequalities are the key instruments for the proof,
	and we prove and use a chain rule for the multi-round case.
Since the honest server model is the weakest attack model, 
	this result implies that the trivial solution is also optimal for any attack models.

\subsubsection{One-server QPIR protocol with prior entanglement in blind setting}
Secondly,
	with prior entanglement between the user and the server,
	we prove that there exists an efficient QPIR protocol on the honest server model.
To be precise, we propose 
	a method to construct a QPIR protocol of communication complexity $O(f(\vM))$ with prior entanglement from C-QPIR protocol of communication complexity $O(f(\vM))$ with prior entanglement.
We construct the QPIR protocol by the combination of C-QPIR and quantum teleportation \cite{BBCJPW93}.
The proposed QPIR protocol inherits the security of C-QPIR.
With this property, on the honest server model with prior entanglement, there exists a QPIR protocol of communication complexity $O(\log \vM)$
	since there exist C-QPIR protocols of communication complexity $O(\log \vM)$ by Kerenidis et al. \cite{KLLGR16}.


\begin{table*}[ht]
\caption{Optimal communication complexity of one-server QPIR} \label{tab:existing_results}

\centering

\begin{subtable}{1\textwidth}
\centering
\caption{One-server QPIR}

\begin{tabular}{|c|c|c|c|}
\hline
Messages &	Server Model	& Optimal communication complexity		& Ref.	\\
\hline
Classical	& 	Honest 		& $ O(\mathrm{poly}\log \vM)$ &\cite{KLLGR16}	\\
Classical	&	Specious	&	$ \Theta(\vM)$ &\cite{BB15}	\\
Quantum (blind)		&	Honest	&	$ \Theta(\vM)$ &[This paper]	\\
Quantum (visible)	&	Honest	&	$ \Theta(\vM)$ (for one-round) &[This paper]\\
\hline
\end{tabular}
\end{subtable}

\vspace{1em}

\begin{subtable}{1\textwidth}
\centering
\caption{One-server QPIR with prior entanglement}
\begin{tabular}{|c|c|c|c|}
\hline
Messages &	Server Model	& Optimal communication complexity		& Ref.	\\
\hline
Classical	& 	Honest 		& $ O(\log \vM)$ 	&\cite{KLLGR16}	\\
Classical	&	Specious	&	$ \Theta(\vM)$ &\cite{ABCGLS19}	\\
Quantum (blind/visible)		&	Honest	&	$ O(\log \vM)$ & [This paper]	\\
\hline
\end{tabular}
\\[0.5em]
{\hfill ($\vM$: Total size of messages) \phantom{abcccccccccccccccccccccccc}}

\end{subtable}

\end{table*}

\subsubsection{Two-server QSPIR protocols for pure states in visible setting}

Lastly, we prove that QSPIR is possible in the visible setting. 
We propose two-server QSPIR protocols for pure states in the visible setting,
	and the protocols work for the different classes of states.
%
The first protocol is for pure qubit states.
This protocol succeeds with probability and 
	requires $8$-qubit average communication and $4$-ebit average entanglement between servers 
	(or $16$-qubit communication without shared entanglement).
The second protocol is for pure {\em qudit} states, where a qudit means a quantum $d$-level system.
This protocol also succeeds with probability and
	requires $4d^d\log d$-qubit average communication and $2d^d\log d$-ebit average entanglement between servers 
	(or $8d^d\log d$-qubit communication without shared entanglement).
The last protocol is for pure states $|\sU \rrangle = \sum_{s,t} u_{st} |s\rangle|t\rangle$ described by commutative unitary matrices $\sU = \sum_{s,t} u_{st} |s\rangle\langle t|$ on qudits.
This protocol succeeds deterministically with $4\log d$-qubit communication and $2\log d$ ebits
	(or $8\log d$-qubit communication without shared entanglement).

The QSPIR protocols have the following important properties. 
First, the protocols delegate most of the quantum information processing to the servers. 
All queries of the user are classical and the only quantum resource required for the user is a measurement apparatus for a fixed basis.
Second, the average quantum communication complexities of these protocols do not scale with the number of message states.
On the other hand, as a trade-off, the size of the queries scales linearly with the number of message states.

The idea of our QSPIR protocols is explained as follows.
In the visible setting, we can prepare the message states $|\psi_\ell\rangle$ by the state preparation unitaries $\sU_\ell$ such that $\sU_\ell|0\rangle = |\psi_\ell\rangle$.
We can decompose any state preparation unitaries $\sU_\ell$ on qudits 
	by phase-shift operations and rotation operations, which are described classically by $2(d-1)$ angles.
Then, with the query structure of classical PIR, 
	the servers collectively encode the phase-shift operations and the rotation operations composing $\sU_k$ 
		into the maximally entangled states, where $k$ is the index of the targeted message state $|\psi_k\rangle$.
By the query structure of classical PIR, this encoding process can be accomplished without leaking $k$ to each server. 
The user receives the maximally entangled states encoded with the phase-shift operations and the rotation operations,
	and finally recovers the targeted message state $|\psi_k\rangle$ by the entanglement-swapping.
The QSPIR protocol for qubit states is done more efficiently since 
	only one rotation operation is necessary for decomposing $\sU_\ell$.
For qudit states for $d \geq 3$,
	the rotation operations for decomposing $\sU_{\ell}$ are non-commutative
	and this non-commutativity leads to the increased complexity.

The remainder of the paper is organized as follows.
Section~\ref{sec:prelim} is the preliminaries of the paper.
Section~\ref{sec:blind-one-opt} derives the lower bound of the communication complexity for one-server QPIR.
In the blind setting, Section~\ref{sec:ent} proposes an efficient one-server QPIR protocol with prior entanglement.
Section~\ref{sec:visible} proposes an efficient QPIR protocol in the visible setting.
Section~\ref{sec:conclusion} is the conclusion of the paper.


\section{Preliminaries} \label{sec:prelim}
We define $[a:b]  = \{a,a+1, \ldots, b\}$ and $[a] = \{1,\ldots, a\}$.
The dimension of a quantum system $X$ is denoted by $|X|$. 
The von Neumann entropy is defined as $H(X) = H(\rho_X) = \Tr \rho_X\log \rho_X$, where $\rho_X$ is the state on the quantum system $X$.

\begin{prop}
The von Neumann entropy satisfies the following properties.

%

 \noindent$(a)$ $H(X) = H(Y)$ if the state on $X\otimes Y$ is a pure state,

 
 \noindent$(b)$ $H(XY) = H(X) + H(Y)$ for product states on $X\otimes Y$,
 
 \noindent$(c)$ Entropy does not change by unitary operations,
 
 \noindent$(d)$ $H(XY) + H(X) \geq H(Y)$,
 
 \noindent$(e)$ $H(\sum_s p_s \rho_s) = \sum_s p_s (H( \rho_s) - \log p_s)$ if $\Tr\rho_s \rho_t = 0$ for any $s\neq t$.
\end{prop}
 
\noindent The property $(d)$ is proved as follows.
Let $Z$ be the reference system in which the state on $XYZ$ is pure. 
Then, $H(XY)+H(X)=H(Z)+H(X)\ge H(XZ)=H(Y)$.
Throughout the paper, we use the symbols $(a)$, $(b)$, $(c)$, $(d)$, $(e)$ to denote which property is used, e.g., $\stackrel{\mathclap{(a)}}{=}$ means that the equality holds from the property $(a)$.

%


\begin{prop}	 \label{prop:pure_uni}
Consider two quantum systems $A$ and $B$.
If any pure states $\rho_{AB}$ and $\rho_{AB}'$ satisfy $\rho_A = \rho_A'$, 
	there exists a unitary $\sU_B$ on $B$ such that $\rho_{AB} = (I_A \otimes \sU_B) \rho_{AB}' (I_A \otimes \sU_B^{*})$.
\end{prop}
\begin{proof}
Diagonalize $\rho_A = \sum_s p_s |s \rangle \langle s |$.
When $\rho_{AB} = |\psi_{AB}\rangle\langle \psi_{AB}|$ and $\rho_{AB}' = |\psi_{AB}'\rangle\langle \psi_{AB}'|$,
	we have
	$|\psi_{AB} \rangle = \sum_s p_s |s,f(s) \rangle$ 
	and
	$|\psi_{AB}' \rangle = \sum_s p_s |s,g(s) \rangle$,
	where $\{|f(s)\rangle\}$ and $\{|g(s)\rangle\}$ are sets of orthonormal vectors of $B$.
Thus, $\sU_B = \sum_s |f(s)\rangle\langle g(s)|$ is the desired unitary.
\end{proof}


For a $d_1\times d_2$ matrix 
	\begin{align}
	\sM = \sum_{s=0}^{d_1-1} \sum_{t=0}^{d_2-1} m_{st} |s\rangle\langle t|  \in \mathbb{C}^{d_1\times d_2},
	\end{align}
	we define 
\begin{align}
|\sM\rrangle = \sum_{s=0}^{d_1-1} \sum_{t=0}^{d_2-1} m_{st} |s\rangle| t\rangle  \in \mathbb{C}^{d_1}\otimes \mathbb{C}^{d_2}.
\end{align}
For $\sA \in \mathbb{C}^{d_1\times d_2}$, $\sB\in\mathbb{C}^{d_1 \times d_1}$, and $\sC \in\mathbb{C}^{d_2 \times d_2}$,
	we have the relation 
\begin{align}
(\sB\otimes \sC^{\top}) | \sA \rrangle =  | \sB\sA\sC \rrangle.
\end{align}


We call a $d$-dimensional system $\mathbb{C}^d$ a {\em qudit}.
Define generalized Pauli matrices and the maximally entangled state on qudits as 
\begin{align}
\sX_d &= \sum_{s=0}^{d-1} |s+1\rangle \langle s|,\\ 
\sZ_d &= \sum_{s=0}^{d-1} \omega^{s} |s\rangle \langle s|,\\
|\sI_d \rrangle &= \frac{1}{\sqrt{d}} \sum_{s=0}^{d-1} |s,s\rangle,
\label{eq:defs_pauli}
\end{align}
where $\omega = \exp(2\pi\ii/ d)$ and $\iota = \sqrt{-1}$.
%
We define 
	the generalized Bell measurements
	\begin{align}
		\mathbf{M}_{\sX\sZ,d} = \{ |\sX^a \sZ^b\rrangle \mid a,b \in [0:d-1] \}. \label{mes}
	\end{align}
If there is no confusion, we denote $\sX_d,\sZ_d, \sI_d, \mathbf{M}_{\sX\sZ,d}$ by $\sX,\sZ, \sI, \mathbf{M}_{\sX\sZ}$.
Let $A, A',B, B'$ be qudits.
If the state on $A\otimes A' \otimes B \otimes B'$ is $|\sA\rrangle\otimes|\sB\rrangle$ 
	and the measurement $\mathbf{M}_{\sX\sZ}$
	is performed on $A' \otimes B'$ with outcome $(a,b) \in [0:d-1]^2$,
	the resultant state is 
	\begin{align}
	|\sA\sX^a\sZ^{-b} \sB^\top\rrangle \in A \otimes B.
		\label{qe:feafterf}
	\end{align}
 

\section{Optimality of Trivial Protocol for One-server QPIR} \label{sec:blind-one-opt}

In this section, we prove that the trivial solution of downloading all messages is optimal 
		for the one-server QPIR in two cases.
We first prove the optimality 
	when the user and the server communicate for one-round in Section~\ref{sec:one}, 
and extends the optimality to the multi-round case in Section~\ref{sec:mult}.
The multi-round case is proved only in the blind setting.

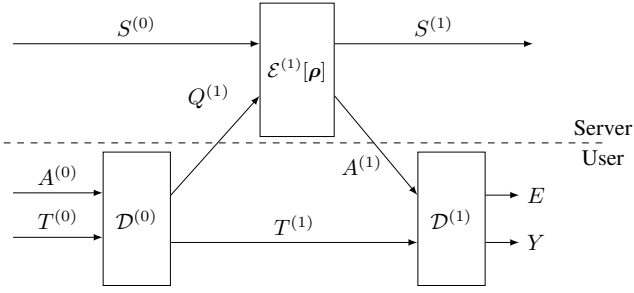
\begin{figure}[t]
\begin{center}
         \resizebox {\linewidth} {!} {
\begin{tikzpicture}[scale=0.5, node distance = 3.3cm, every text node part/.style={align=center}, auto]


	\node (enc1) {};
	\node [block, minimum height=6em, below right=-2.2em and 11em of enc1] (enc2) {\small $\cE^{(1)}[\bm{\rho}]$};
	\node [block, minimum height=6em, below right=4.5em and 4em of enc1] (dec1) {$\mathcal{D}^{(0)}$};
	\node [block, minimum height=6em, right=11em of dec1] (dec2) {$\mathcal{D}^{(1)}$};

	\node [below=6em of enc1] (dd1) {};
	\node [below=8em of enc1] (dd2) {};
	\path [line] (dd1) --node[above]{$A^{(0)}$} (dd1-|dec1.145);
	\path [line] (dd2) -- node[above]{$T^{(0)}$} (dd2-|dec1.215);
	

	\node [left=0em of enc1.145] (start) {};

	\node [below left=4em and -1em of start] (lefts) {};
	\node [right=27em of lefts] (rights) {};
	\path [line,-,dashed] (lefts) -- node[below,pos=0.99]{User} node[above,pos=0.99]{Server} (rights);

	\path [line] (enc1) -- node[above]{$S^{(0)}$} (enc1 -|enc2.145);
	\path [line] (enc2.325) -- node[below,pos=0.7,left=0.1em]{$A^{(1)}$} (dec2.145);

	\path [line] (dec1.35) --node[above,above left=1.3em and -1em]{$Q^{(1)}$} (enc2.215);
	\path [line] (dec1.325) -- node[above]{$T^{(1)}$} (dec2.215);

	\path [line] (enc2.35) --node[above]{$S^{(1)}$} ++(0:17.7em);
	\path [line] (dec2.325) -- ++(0:3em) node[right] {$Y$};
	\path [line] (dec2.35) -- ++(0:3em) node[right] {$E$};

\end{tikzpicture}
 }
\caption{One-round QPIR protocol in the visible setting.}		\label{fig:one-server} 
\end{center}
\end{figure}

\subsection{Optimality of trivial protocol for one-server one-round QPIR} \label{sec:one}

In this subsection, we prove that the optimal communication complexity for one-server one-round QPIR is $\vM$, where $\vM$ is the total number of qubits in the message states.
The achievability part is proved by the trivial solution, which can be implemented in the blind setting, 
	and 
	we will prove the tight lower bound in the visible setting.
Since any protocol in the blind setting can be regarded as a protocol in the visible setting, 
	our tight lower bound in the visible setting proves that the trivial solution is optimal for both blind and visible settings.

	{
We formally describe one-server one-round QPIR protocols in the visible setting with Figure~\ref{fig:one-server} as follows.
The message states are given as arbitrary $\vF$ states 
	$\bm{\rho} = (\rho_1,\ldots,\rho_{\vF})$ on quantum systems $X_1, \ldots, X_{\vF}$, respectively,
The server has the descriptions of all messages states $\bm{\rho}$.
The user chooses the index of the targeted message $K\in[\vF]$, i.e., $\rho_k$ is the targeted quantum state when $K=k$.
We assume that the user and the server contain local quantum registers, respectively, 
	so that all local operations are written as unitary operations.
A QPIR protocol $\Phi$ is described by three unitary maps $(\cD^{(0)}, \cD^{(1)}, \cE^{(1)})$ in the following steps.

\begin{enumerate}[leftmargin=1.5em]
\item \textbf{Query}: 
	When $K=k$,
		the user prepares the initial state as $|k\rangle \otimes |0\rangle \in A^{(0)} \otimes T^{(0)}$,
	applies a unitary map $\cD^{(0)}$ from $A^{(0)} \otimes T^{(0)}$ to $Q^{(1)}\otimes T^{(1)}$,
	and sends $Q^{(1)}$ to the server.

\item \textbf{Answer}:
The server prepares the initial state $|0\rangle \in S^{(0)}$,
	applies a unitary map $\cE^{(1)}[{\bm{\rho}}]$ 
		from $S^{(0)} \otimes Q^{(1)}$ to $A^{(1)}\otimes S^{(1)}$,
	and returns to the user the system $A^{(1)}$.

\item \textbf{Reconstruction}:
	The user applies a unitary map $\cD^{(1)}$ from $A^{(1)}\otimes T^{(1)}$ to $Y\otimes E$,
	and outputs the state on $Y$ as the protocol output.
\end{enumerate}

The input-output relation 
	$\Lambda_{\Phi}(k,{\bm{\rho}})$ 
		of the QPIR protocol $\Phi$ is written with 
		a CPTP $\Gamma_{\Phi,{\bm{\rho}}}$ as
	\begin{align*}
	\Lambda_{\Phi}(k,{\bm{\rho}}) 
	&= \Gamma_{\Phi,{\bm{\rho}}} (|k\rangle\langle k|)\\
	&= \Tr_{S^{(1)},E} \cD^{(1)} \circ \cE^{(1)}[{\bm{\rho}}] \circ \cD^{(0)} (|k\rangle\langle k|\otimes |0\rangle\langle 0|).
	\end{align*}
The QPIR protocol $\Phi$ should satisfy the following conditions.
\begin{itemize}[leftmargin=1.5em]
\item \textbf{Correctness}: 
	When the output state on $Y$ with the target index $K=k$ is denoted by $\rho_Y^{k}$,
	the correctness is 
	\begin{align}
	\rho_Y^k = \rho_{k}
	 \label{eq:corr_one}
	\end{align}
	for any $k$ and any message states $\bm{\rho}$.

\item \textbf{User secrecy}: 
	When the state on $Q^{(1)}$ with the target index $K=k$ is denoted by 
	$\rho_{Q^{(1)}}^k,$
	the user secrecy is 
		\begin{align}
		\rho_{Q^{(1)}}^k
		&= \rho_{Q^{(1)}}^{\ell}
		\label{eq:sec1cc}
		\end{align}
		for any $k,\ell$.
\end{itemize}
}

We evaluate the efficiency of a QPIR protocol $\Phi$ by the communication complexity 
	$\CC \coloneqq \log |Q^{(1)}| +  \log |A^{(1)}|$, which is the whole dimension of uploaded and downloaded systems.
The communication complexity of the trivial solution of downloading all states
	is $\sum_{\ell=1}^{\vF} \log |X_\ell|$. 
The following theorem proves the optimality of the trivial solution. 
\begin{theo} \label{theo:one}
For any one-server one-round QPIR protocol $\Phi$, 
	the communication complexity $\CC$ is lower bounded as
\begin{align}
\CC \ge \sum_{\ell=1}^{\vF} \log |X_\ell|, \label{eq:oneee}
\end{align}
where $X_\ell$ is the system of the $\ell$-th message $\rho_\ell$.
\end{theo}

For the proof of Theorem~\ref{theo:one}, we prepare the following notations and lemma.
Given the description of states ${\bm{\rho}}$ and the target index $k$,
	we denote the state on the system $A^{(1)}\otimes T^{(1)}$ by $\rho_{A^{(1)}T^{(1)}}^{k,{\bm{\rho}}}$.
{If the message states are pure states $|\phi_1\rangle, \ldots , |\phi_{\vF}\rangle$,
	we denote $\rho_{A^{(1)}T^{(1)}}^{k,{\bm{\rho}}}$ by $\rho_{A^{(1)}T^{(1)}}^{k,|\phi_1,\ldots ,\phi_{\vF}\rangle}$.}
Let $d_{\ell} \coloneqq |X_{\ell}|$
	and $d \coloneqq \prod_{\ell=1}^{\vF} d_{\ell}$. 

\begin{lemm}	\label{lemm:pruni}
{Suppose \eqref{eq:sec1cc}.}
For any $k,\ell \in [\vF]$,
	there exists a unitary $\sU_{T^{(1)}}^{k\to \ell}$ on $T^{(1)}$ such that 
	\begin{align}
	\sU_{T^{(1)}}^{k\to \ell}\rho_{A^{(1)}T^{(1)}}^{k,{\bm{\rho}}}(\sU_{T^{(1)}}^{k\to \ell})^{\dagger}
	 &= \rho_{A^{(1)}T^{(1)}}^{\ell,{\bm{\rho}}}
	 \label{lemeq:uni}
	\end{align}
	for any message states $\bm{\rho}$.
\end{lemm}
\begin{proof}
Let $\rho_{Q^{(1)} \otimes T^{(1)}}^{k}$ on $Q^{(1)} \otimes T^{(1)}$ be the state on $Q^{(1)} \otimes T^{(1)}$ for user's input $k\in[\vF]$.
From the secrecy condition \eqref{eq:sec1cc},
 	the state on $Q^{(1)}$ does not depend on the value of $k$.
Since the state on $Q^{(1)} \otimes T^{(1)}$ is a pure state,
	there exists a unitary $\sU_{T^{(1)}}^{k\to \ell}$ on $T^{(1)}$ 
	that maps  $\rho_{Q^{(1)} \otimes T^{(1)}}^{k}$ to $\rho_{Q^{(1)} \otimes T^{(1)}}^{\ell}$ 
		by Proposition~\ref{prop:pure_uni}.
This unitary $\sU_{T^{(1)}}^{k\to \ell}$ does not depend on the message states $\bm{\rho}$.
Since the server's operation is not applied to $T^{(1)}$,
	the same unitary $\sU_{T^{(1)}}^{k\to \ell}$ satisfies \eqref{lemeq:uni}.
%
\end{proof}

\begin{proof}[Proof of Theorem~\ref{theo:one}]
{First, we prove that 
	\begin{align}
	\px*{
		\rho_{A^{(1)}T^{(1)}}^{1,{|s_1,\ldots , s_{\vF}\rangle }}
		\mid
		s_\ell \in [0:d_{\ell}-1], \ \ell \in [\vF]
	}
	\label{eq:states_orth}
	\end{align}
	is a set composed of orthogonal states.
Let $(s_1,\ldots,s_{\vF})$, $(s_1',\ldots, s_{\vF}')$ be any two different strings
	and $k\in[\vF]$ be an index satisfying $s_k \neq s_k'$.
We have 
	\begin{align}
	\Tr 
		\rho_{A^{(1)}T^{(1)}}^{k,{|s_1,\ldots , s_{\vF}\rangle}}
		\rho_{A^{(1)}T^{(1)}}^{k,{|s_1',\ldots , s_{\vF}'\rangle}}
		=
		0,
		\label{eq:labb13}
	\end{align}
	since the correctness condition \eqref{eq:corr_one} guarantees that the two states can be made orthogonal by the decoding unitary map $\cD^{(1)}$.
Furthermore, 
	for the same $(s_1,\ldots,s_{\vF})$ and $(s_1',\ldots, s_{\vF}')$,
	Eq.~\eqref{eq:labb13} implies that
	\begin{align}
	\Tr 
		\rho_{A^{(1)}T^{(1)}}^{1,{|s_1,\ldots , s_{\vF}\rangle}}
		\rho_{A^{(1)}T^{(1)}}^{1,{|s_1',\ldots , s_{\vF}'\rangle}}
		=
		0,
		\label{eq:equiv_ort}
	\end{align}
	because there exists a unitary map $\sU_{T^{(1)}}^{1\to k}$ on ${T^{(1)}}$ such that 
	\begin{align}
	\sU_{T^{(1)}}^{1\to k}\rho_{A^{(1)}T^{(1)}}^{1,{|s_1,\ldots , s_{\vF}\rangle}}(\sU_{T^{(1)}}^{1\to k})^{\dagger}
	 &= \rho_{A^{(1)}T^{(1)}}^{k,{|s_1,\ldots , s_{\vF}\rangle}}, \\
	\sU_{T^{(1)}}^{1\to k}\rho_{A^{(1)}T^{(1)}}^{1,{|s_1',\ldots , s_{\vF}'\rangle}}(\sU_{T^{(1)}}^{1\to k})^{\dagger}
	 &= \rho_{A^{(1)}T^{(1)}}^{k,{|s_1',\ldots , s_{\vF}'\rangle}}
	\end{align}
	from Lemma~\ref{lemm:pruni}. 
Since \eqref{eq:equiv_ort} holds for any different $(s_1,\ldots, s_{\vF})$ and $(s_1',\ldots, s_{\vF}')$,
	the set \eqref{eq:states_orth} is composed of orthogonal states.}


Next, 
	we consider the case where
	the user's input is $k=1$ 
	and
	the server applies the CPTP map
	\begin{align}
	\sum_{s_1,\ldots,s_{\vF}}
	\frac{1}{d} \cE^{(1)}[|s_1,\ldots,s_{\vF}\rangle],
	\end{align}
	which is the uniform mixture of encoding maps.
Then, the state on $A^{(1)}\otimes R^{(1)}$ is 
	\begin{align}
	\rho_{\ast}  
	&\coloneqq \sum_{s_1,\ldots,s_{\vF}}
	\frac{1}{d} \cE^{(1)}[|s_1,\ldots,s_{\vF}\rangle] \circ \cD^{(0)} \paren*{ |1\rangle\langle 1| \otimes |0\rangle\langle 0|} 
	\\
	&= 
		\sum_{s_1,\ldots,s_{\vF}} 
		\frac{1}{d} 
		\rho_{A^{(1)}T^{(1)}}^{1,{|s_1,\ldots , s_{\vF}\rangle }},
	\end{align}
	and the von Neumann entropy of $\rho_{\ast}$ is 
	\begin{align}
	H(A^{(1)} T^{(1)})_{\rho_{\ast}}
	&=
	H(\rho_{\ast}) 
	\\
	&\stackrel{\mathclap{(e)}}{=} 
		\sum_{s_1,\ldots,s_{\vF}} 
		\frac{1}{d}  H(\rho_{A^{(1)}T^{(1)}}^{1,{|s_1,\ldots , s_{\vF}\rangle }} )
		+ \log d
		\label{eq:eq_disjoint}
	\\
	&\geq \log d.
		\label{ineq:low_cctp}
	\end{align}
	where \eqref{eq:eq_disjoint} follows from the orthogonality of the states in \eqref{eq:states_orth}.
Furthermore,
	we have 
	\begin{align}
	H(A^{(1)}T^{(1)})_{\rho_\ast} &\le H(A^{(1)})_{\rho_\ast} + H(T^{(1)})_{\rho_\ast}\\
			&\stackrel{\mathclap{(a)}}{=}  H(A^{(1)})_{\rho_\ast} + H(Q^{(1)})_{\rho_{Q^{(1)}}^1}\\
			&\le \log |A^{(1)}| + \log |Q^{(1)}| \\
			&= \CC.
	\label{ineq:lowccone}
	\end{align}
Combining \eqref{ineq:low_cctp} and \eqref{ineq:lowccone}, we obtain the desired inequality \eqref{eq:oneee}.
%
%
\end{proof}

\begin{figure*}[t]
\begin{center}
\begin{tikzpicture}[scale=0.5, node distance = 3.3cm, every text node part/.style={align=center}, auto]

	
	\node [block, minimum height=6em] (enc1) {$\mathcal{E}^{(1)}$};
	\node [block, minimum height=6em, right=11em of enc1] (enc2) {$\mathcal{E}^{(2)}$};
	\node [block, minimum height=6em, below right=2.5em and 4em of enc1] (dec1) {$\cD^{(1)}$};
	\node [block, minimum height=6em, right=11em of dec1] (dec2) {$\cD^{(2)}$};

	\node [left=7em of enc1.145] (start) {\phantom{$X_{[\vF]}$}};
	\node [above=0.3em of start] (XX) {\phantom{$X_{[\vF]}$}};
	\path [line] (XX.east) --++(4em, 4em) -- node[above, pos=0.07] {$R_{[\vF]}$} ++(0:62em);
	\path [line] (XX.east) --++(4em,-4em) --node[above]{$X_{[\vF]}$}  (start -| enc1.west);
	
	\node [below left=5.7em and 7em of enc1.145] (K)  {$K$};
	\path [line] (K.east) --++(4em, 8.5em) -- node[above] {$Q^{(1)}$} (enc1.215);
	\path [line] (K.east) --++(4em,-8.5em) --node[above, pos=0.16]{$T^{(1)}$}  (dec1.215);

	\node [below left=3.5em and -1em of start] (lefts) {};
	\node [right=35em of lefts] (rights) {};
	\path [line,-,dashed] (lefts) -- node[below,pos=0.99]{User} node[above,pos=0.99]{Server} (rights);

	\path [line] (enc1.35) --node[above]{$S^{(1)}$} (enc2.145);
	\path [line] (enc1.325) -- node[below,pos=0.7,left=0.1em]{$A^{(1)}$} (dec1.145);
	\path [line] (enc2.325) -- node[below,pos=0.7,left=0.1em]{$A^{(2)}$} (dec2.145);

	\path [line] (dec1.35) --node[above,above left=1.3em and -1em]{$Q^{(2)}$} (enc2.215);
	\path [line] (dec1.325) -- node[above]{$T^{(2)}$} (dec2.215);

	\path [line] (enc2.35) --node[above]{$S^{(2)}$} ++(0:17.7em);
	\path [line] (dec2.325) -- ++(0:3em) node[right] {$Y$};
	\path [line] (dec2.35) -- ++(0:3em) node[right] {$E$};

\end{tikzpicture}
\caption{$2$-round QPIR protocol in the blind setting.}		\label{fig:flow}
\end{center}
\end{figure*}
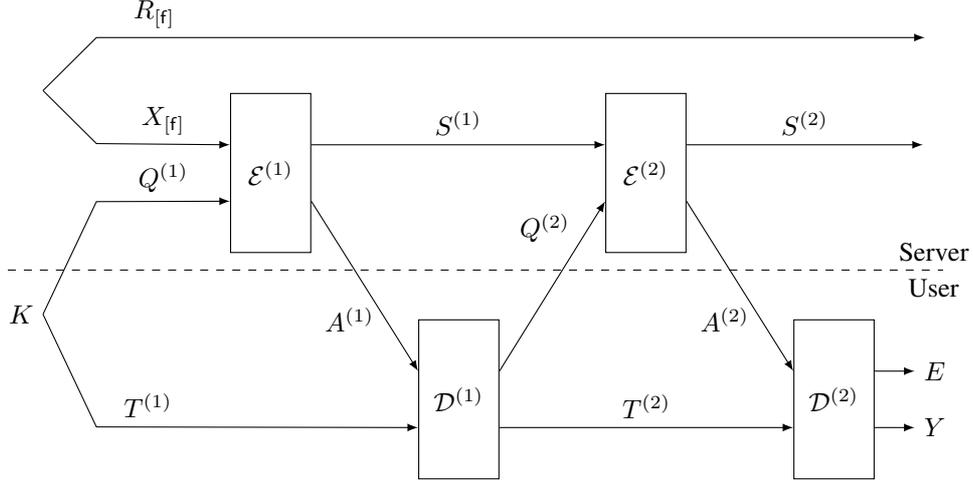

\subsection{Optimality of trivial protocol for one-server multi-round QPIR in blind setting} \label{sec:mult}
{
In this subsection, we extend the the optimality of the trivial solution to the case where the user and the server communicate multiple rounds in the blind setting.
To be precise, we define the $\vR$-round QPIR protocol in the blind setting as follows ($2$-round protocol is depicted in Figure~\ref{fig:flow}).

The message states are given as arbitrary $\vF$ states 
	$\rho_1\otimes \cdots\otimes \rho_{\vF}$ on $S^{(0)} = X_{1}\otimes \cdots \otimes X_{\vF}$, where each of $\rho_\ell$ is purified in $X_\ell\otimes  R_\ell$. 
The server contains the system $S^{(0)}$.
The user chooses the index of the targeted message $K\in[\vF]$, i.e., $\rho_k$ is the targeted quantum state when $K=k$.
When $K=k$, the user prepares the initial state as $|k\rangle \otimes |0\rangle \in A^{(0)} \otimes T^{(0)}$,
We assume that the user and the server contain local quantum registers, respectively, which enable that all local operations are written as unitary operations.
A QPIR protocol $\Phi$ is described by unitary maps $\cD^{(0)},\ldots,\cD^{(\vR)}, \cE^{(1)},\ldots,\cE^{(\vR)}$ in the following steps.

\begin{enumerate}[leftmargin=1.5em]
\item \textbf{Query}: 
For all $i \in [\vR]$,
	the user 
	applies a unitary map $\cD^{(i-1)}$ from $A^{(i-1)}\otimes T^{(i-1)}$ to $Q^{(i)} \otimes T^{(i)}$,
	and sends $Q^{(i)}$ to the sender.

\item \textbf{Answer}:
For all $i \in [\vR]$,
	the server applies a unitary map $\mathcal{E}^{(i)}$ from $Q^{(i)} \otimes S^{(i-1)}$ to $A^{(i)}\otimes S^{(i)} $
	and sends $A^{(i)}$ to the user.

\item \textbf{Reconstruction}:
	The user applies $\cD^{(\vR)}$ from $A^{(\vR)}\otimes T^{(\vR)}$ to $Y \otimes E$, 
	and outputs the state on $Y$ as the protocol output.
	
\end{enumerate}

The input-output relation $\Lambda_{\Phi}$ of the protocol $\Phi$ is written with a CPTP $\Gamma_{\Phi,k}$ from $X_{[\vF]}$ to $Y$ as
	\begin{align*}
	&\Lambda_{\Phi} (k,\rho_1,\ldots,\rho_{\vF}) 
	= \Gamma_{\Phi,k} (\rho_{[\vF]})\\
	&= \Tr_{S^{(\vR)},E}  \cD \ast \cE ( \rho_{[\vF]} \otimes \cD^{(0)} (|k\rangle\langle k|  \otimes |0\rangle\langle 0|) ),
	\end{align*}
	where $\cD \ast \cE = ( \cD^{(\vR)} \circ \cE^{(\vR)} )\circ\cdots \circ ( \cD^{(1)} \circ \cE^{(1)} ) $.
The QPIR protocol $\Phi$ should satisfy the following conditions.
\begin{itemize}[leftmargin=1.5em]
\item \textbf{Correctness}: 
	When $|\psi_k\rangle \langle \psi_k |$ denotes a purification of $\rho_k$ with the reference system $R_k$,
	the correctness is 
	$$
	\Gamma_{\Phi,k}\otimes \id_{R_k}(\rho_{[\vF]\setminus \{k\}} \otimes |\psi_k\rangle \langle \psi_k |)
	 = |\psi_k\rangle \langle \psi_k |
	$$
	for any $K=k$ and any state $\rho_{[\vF]}$.

\item \textbf{User secrecy}: 
	When the state on $S^{(i-1)} \otimes Q^{(i)}$ with the target index $K=k$ is denoted by 
	$\rho_{S^{(i-1)} Q^{(i)}}(k),$
	the user secrecy is 
		\begin{align}
		\rho_{S^{(i-1)} Q^{(i)}}(k)
		&= \rho_{S^{(i-1)} Q^{(i)}}(k')
		\label{eq:sec2cc}
		\end{align}
		for any $k,k',i$.
\end{itemize}
The communication complexity of the one-server multi-round QPIR is written as $\CC = \sum_{i=1}^{\vR}  \log |Q^{(i)}| +  \log |A^{(i)}|$.
}

\begin{theo} \label{theo:multiround}
For any one-server multi-round QPIR protocol $\Phi$ in the blind setting,
the communication complexity $\CC$ is lower bounded by $\sum_{\ell=1}^{\vF} \log |X_\ell|$, where $X_\ell$ is the system of the $\ell$-th message $\rho_\ell$.
\end{theo}


For the proof of Theorem \ref{theo:multiround},
	we prepare the following lemmas.
\begin{lemm} \label{lemm:recurs}
$H(A^{(i)}) +  H(Q^{(i+1)})  \geq H(T^{(i+1)} ) - H(T^{(i)})$.
\end{lemm}
\begin{proof}
Lemma \ref{lemm:recurs} is shown by the relation
\begin{align*}
 &H(A^{(i)}) + H(T^{(i)}) +  H(Q^{(i+1)}) \\
 &\stackrel{\mathclap{(b)}}{\geq} H(A^{(i)} T^{(i)} ) +  H(Q^{(i+1)}) \\
 &\stackrel{\mathclap{(c)}}{=} H(Q^{(i+1)} T^{(i+1)} ) +  H(Q^{(i+1)}) \\
 &\stackrel{\mathclap{(d)}}{\ge} H(T^{(i+1)} ).
%
\end{align*}
\end{proof}


\begin{lemm} \label{lemm:prrrd}
$H( R_{[\vF]} S^{(\vR)}) = H(R_{[\vF]}) + H(S^{(\vR)}).$ 
\end{lemm}
\begin{proof}
%
Given the user's input $k$, the state on $R_k\otimes Y$ is a pure state, 
and therefore,
	$R_k$ is independent of any system except for $Y$.
Thus, we have
\begin{align}	
H( R_{[\vF]} S^{(\vR)} ) =  H( R_{[\vF]\setminus \{k\} } S^{(\vR)}) + H(R_k)  \label{eq:222f}
\end{align}
	for any $k$.
Note that the state on $R_{[\vF]} \otimes S^{(\vR)}$ does not depend on $k$ due to the secrecy condition \eqref{eq:sec2cc}.
Thus, applying \eqref{eq:222f} recursively for all $k$,
	we have
	\begin{align}
	H( R_{[\vF]} S^{(\vR)} ) 
		&=  H(R_{\vF})  + H( R_{[\vF-1] } S^{(\vR)}) \\
		&=  H(R_{\vF-1}R_{\vF}) + H( R_{[\vF-2] } S^{(\vR)})  
		= \cdots	\\
		&= H(R_{[\vF]}) + H(S^{(\vR)}). 
	\end{align}
\end{proof}

\begin{proof}[Proof of Theorem~\ref{theo:multiround}]
From Lemmas~\ref{lemm:recurs} and \ref{lemm:prrrd}, we derive the following inequalities:
\begin{align}
	&\CC \geq \sum_{i=1}^{\vR} H(A^{(i)}) + H(Q^{(i)})\\
	 &= H(A^{(\vR)}) + H(Q^{(1)}) + \sum_{i=1}^{\vR-1} H(A^{(i)}) + H(Q^{(i+1)})	\\
	 &\geq H(A^{(\vR)}) + H(Q^{(1)}) + H(T^{(\vR)}) - H(T^{(1)}) \label{eq:lemmapp}\\
	 &\stackrel{\mathclap{(a)}}{=} H(A^{(\vR)}) + H(T^{(\vR)}) \\ 
	 &\stackrel{\mathclap{(b)}}{\geq} H(A^{(\vR)} T^{(\vR)}) \stackrel{\mathclap{(a)}}{=} H(R_{[\vF]} S^{(\vR)}) \\
	 &= H(R_{[\vF]}) +  H(S^{(\vR)}) \label{eq:samewa} \\
	 &\geq H(R_{[\vF]}) = H(X_{[\vF]}),
\end{align}
where \eqref{eq:lemmapp} is obtained by applying Lemma~\ref{lemm:recurs} for all $i=1,\ldots, \vR-1$,
and \eqref{eq:samewa} is from Lemma~\ref{lemm:prrrd}.
Taking the maximum of $H(X_{[\vF]})$ over all states $\rho_1,\ldots,\rho_{\vF}$,
	we obtain $\CC \geq  \sum_{\ell=1}^{\vF} \log |X_\ell|$.
\end{proof}

\section{QPIR Protocol with Prior Entanglement in Blind Setting} \label{sec:ent}
In the previous section, we proved that the trivial solution is optimal for one-server QPIR.
In this section, we show that if we allow shared entanglement between the user and the servers,
	we can construct a QPIR protocol with lower communication complexity than the trivial solution.


Let $\vM = \sum_{\ell=1}^{\vF} \log |X_\ell|$ be the size of all messages.
To measure the amount of the prior entanglement,
	we count sharing one copy of $|\sI_2\rrangle = (1/\sqrt{2})(|00\rangle + |11\rangle)$ as an {\em ebit}. 
Accordingly, we count sharing the state $|\sI_d\rrangle \in \mathbb{C}^d \otimes \mathbb{C}^d$ as $\log d$ ebits.

%
%

\begin{theo} \label{theo:prot} 
Suppose there exists a C-QPIR protocol with communication complexity $O(f(\vM))$
	when $O(g(\vM))$-ebit prior entanglement is shared between the user and the server.
Then, in the blind setting, there exists a QPIR protocol for quantum messages with communication complexity $O(f(\vM))$
	when $O(\vM+g(\vM))$-ebit prior entanglement is shared between the user and the server.
\end{theo}


The protocol satisfying Theorem \ref{theo:prot} is a simple combination of quantum teleportation \cite{BBCJPW93} and any C-QPIR protocol.
For the description of the protocol, we use the generalized Pauli operators and maximally entangled state for $d$-dimensional systems defined in \eqref{eq:defs_pauli}.
We construct the QPIR protocol satisfying Theorem \ref{theo:prot} as follows.
\begin{prot} \label{prot:ent}
Let $\Phi_{\mathrm{cl}}$ be a C-QPIR protocol 
	and $d_1,\ldots, d_{\vF}$ be the size of the $\vF$ classical messages. 
%
From this protocol, we construct a QPIR protocol for quantum messages as follows.

Let $X_1,\ldots, X_\vF$ be the quantum systems with dimensions $d_1,\ldots, d_{\vF}$, respectively,
	and 
 $\rho_1,\ldots, \rho_{\vF}$ be the quantum message states on systems $X_1,\ldots, X_\vF$.
The user and the server share the maximally entangled states $|\sI_{d_\ell}\rrangle$, defined in \eqref{eq:defs_pauli}, on ${Y_\ell\otimes Y_\ell'}$ for all $\ell\in[\vF]$,
	where $Y_{[\vF]}$ and $Y_{[\vF]}'$ are possessed by the user and the server, respectively.
	
The user and the server perform the following steps.
\begin{enumerate}[leftmargin=1.5em]
\item 
		For all $\ell\in[\vF]$, the server performs the generalized Bell measurement
		$\mathbf{M}_{\sX\sZ,d_\ell}$, defined in \eqref{mes},
		on $X_\ell\otimes Y_\ell'$, where the measurement outcome is written as $m_\ell = (a_\ell, b_\ell)\in [0:  d_\ell-1]^2$.

\item 
	The user and the server perform the C-QPIR protocol $\Phi_{\mathrm{cl}}$ to retrieve $m_k = (a_k,b_k)$.
%
\item The user recovers the $k$-th message $\rho_k$ by applying $\sX_{d_k}^{-a_k}\sZ_{d_k}^{b_k} $ on $Y_k$.
\QEDA
\end{enumerate}
\end{prot}
The correctness of the protocol is guaranteed by correctness of the teleportation protocol and the C-QPIR protocol $\Phi_{\mathrm{cl}}$.
When the $\ell$-th message state is prepared as $\rho_\ell$ and its purification $|\phi_\ell\rangle$ is denoted with the reference system $R_\ell$,
	after Step~1, the states on $R_\ell\otimes Y_\ell$ is 
	\begin{align}
	( \sI\otimes  \sX_{d_\ell}^{a_\ell} \sZ_{d_\ell}^{-b_\ell}) |\phi_\ell\rangle
	\end{align}
for all $\ell\in[\vF]$.
Thus after Step~3, the target state $|{\phi_k}\rangle$ is recovered in $R_k\otimes Y_k$.


To analyze the secrecy of Protocol~\ref{prot:ent}, 
	note that only Step~2 has the communication between the user and the server.
Thus the secrecy of Protocol~\ref{prot:ent} is guaranteed by the secrecy of the underlying protocol $\Phi_{\mathrm{cl}}$.
In the honest server model, Kerenidis et al. \cite{KLLGR16} proposed a C-QPIR protocol with communication complexity $O(\log\vM)$ and prior entanglement $O(\vM)$.
Therefore, we obtain the following corollary. 
\begin{coro}
On the honest server model in the blind setting,
	there exists a QPIR protocol for quantum messages with communication complexity $O(\log\vM)$
	and prior entanglement $O(\vM)$.
\end{coro}

One property of Protocol~\ref{prot:ent} is that all other states in the server are destroyed at Step~1.
This is a disadvantage for the server but an advantage for the user since the user can retrieve other states $\rho_\ell$ by retrieving classical information $m_\ell \in [0:d_{\ell}-1]^2$.

\begin{table*}[ht]
\begin{center}
\caption{Cost of Two-Server QSPIR Protocols in Visible Setting} \label{tab:vis}
\begin{tabular}{|c|c|c|c|c|}
\hline
						&	Message States				&	Classical Communication	&	Quantum Communication	&	Prior Entanglement		\\
\hline
Protocol~\ref{prot:npe}	&	Pure qubit states 		&	$2\vF$ bits			&	$8$ qubits			&	$4$ ebits	\\
\hline
Protocol~\ref{prot:dlev}&	Pure qudit states		&	$2\vF$ bits		&	$4d^d\log d $ qubits	&	$2d^d \log d$ ebits	\\
\hline
Protocol~\ref{prot:wpe}	&	Commutative unitary		&	$2\vF$ bits				&	$4\log d$ qubits			&	$2\log d$ ebits\\
\hline
\end{tabular}
\end{center}
\end{table*}

\section{Two-Server Symmetric QPIR Protocols in Visible Setting} \label{sec:visible}

In this section, we propose two-server one-round QSPIR protocols with classical query in the visible setting.
%
%
In the multi-server model, the servers cannot communicate with each other.
In the visible setting, the server has the description of quantum states instead of the states. 
With a description of a state $\rho$, the server may generate multiple copies of $\rho$, without limitation of the no-cloning theorem, and apply quantum operations depending on the description of $\rho$.
Our protocols in this section are symmetric QPIR (QSPIR) in which the user only obtains the information of the targeted message state.
Thus, our protocols prove that the QSPIR is possible in the visible setting.

\subsection{Definition and main theorems}

We propose three QSPIR protocols in the visible setting for pure states.
The communication complexity of the protocols are summarized in Table~\ref{tab:vis}.
{Even if the QPIR's trivial solution of downloading all states is not a QSPIR protocol, 
	we can evaluate the efficiency of our QSPIR protocols by comparison with the communication complexities of the QPIR's trivial solution.}
Compared to the communication complexity of the QPIR's trivial solution is $\vF\log d$ qubits,
	the quantum communication complexities of our protocols do not scale with the number of messages $\vF$ 
		but scale with the dimension of the systems in which the message states are prepared.
Thus, when the number of messages $\vF$ is sufficiently greater than the dimension, 
	the proposed protocols are more efficient than the QPIR's trivial solution of downloading all messages.

Our QSPIR protocols are included in the class of $\vN$-server QSPIR protocols in the visible setting with the classical queries described as follows.
Let $\pureset{\mathbb{C}^d}$ be the set of pure states in $\mathbb{C}^d$.
We formally define a QSPIR protocol $\Phi$ for $\mathcal{P}\subset \pureset{\mathbb{C}^d}$.
The message states are given as arbitrary $\vF$ states 
	$\bm{\psi} = (|\psi_1\rangle , \ldots, |\psi_{\vF}\rangle) \in \mathcal{P}^{\vF}$
and each server has the descriptions of all messages states. 
The user chooses the index of the targeted message $K\in[\vF]$, i.e., $|\psi_k\rangle$ is the targeted quantum state when $K=k$.
A protocol $\Phi$ constructed by the following steps.
\begin{enumerate}[leftmargin=1.5em]
\item \textbf{Query}: The user randomly encodes $K$ as classical queries $Q=(Q_1,\ldots, Q_{\vN}) \in \cQ_1\times \cdots \times \cQ_{\vN}$ and sends $Q_j$ to the $j$-th server.
That is, 
	the variable $Q$ is subject to the conditional distribution $p_{Q|K=k}$, which is chosen by the user.
	
\item \textbf{Entanglement Sharing}:
	Independently of the queries, the servers share an entangled state $\rho_{\mathrm{init}}$ on $S_1\otimes \cdots \otimes S_{\vN}$,
		where $S_j$ is contained in the $j$-th server.
		
\item \textbf{Answer}:
	For all $j\in[\vN]$,
		the $j$-th server applies a CPTP map  $\cE_j[\bm{\psi},q_j]$ from $S_j$ to $A_j$
		and sends $A_j$ to the user when $Q_j=q_j$.
		
\item \textbf{Reconstruction}:
	{
	Let $A\coloneqq A_1\otimes \cdots \otimes A_{\vN}$ and decompose $A = Y \otimes Y'$.
	The user performs a measurement by a POVM $\{\sM, \sI-\sM\}$ on $Y \subset A$.
	If $\sM$ is measured, the targeted state $|\psi_k\rangle$ is recovered correctly in $Y'$.
	If the reconstruction fails, i.e. $\sI-\sM$ is measured, repeat from Step 2.}
	
\end{enumerate}
The protocol $\Phi$ should satisfy the following conditions.
\begin{itemize}[leftmargin=1.5em]
\item \textbf{Correctness}: At each execution of Step 4, 
	the user recovers the targeted state $|\psi_k\rangle$ with positive probability $p$.
	
\item \textbf{User secrecy}: 
	The user secrecy is 
		\begin{align}
		I( Q_j ; K ) = 0 \quad \forall j
		\end{align}
		for any distribution of $K$.

\item \textbf{Server secrecy}: 
	Let $\rho_{A} (k, \bm{\psi})$ be the state on $A$ with the target index $k$ and the message states 
	$\bm{\psi} = (|\psi_1\rangle , \ldots, |\psi_{\vF}\rangle)$.
	The server secrecy is 
		\begin{align*}
		&\rho_{A}(k, \bm{\psi}) 
		=
		\rho_{A}(k, \bm{\phi})
		\end{align*}
		for {any $k\in[\vF]$ and} any states $\bm{\psi} = (|\psi_1\rangle , \ldots, |\psi_{\vF}\rangle), \bm{\phi} = (|\phi_1\rangle , \ldots, |\phi_{\vF}\rangle) \in \cP^{\vF}$ such that $|\psi_k\rangle = |\phi_{k}\rangle$.
%
\end{itemize}
The average classical and quantum communication complexities are defined as
\begin{align}
\CC_{c} &= \sum_{j=1}^{\vN} \log |\cQ_j|, \\
\CC_{q} &= \frac{1}{p}\sum_{j=1}^{\vN} \log |A_j|,
\end{align}
since the classical queries are uploaded once and 
	the quantum answers are downloaded $1/p$ times on average.

Our protocols will be defined for $\vN=2$ servers with the maximally entanglement states as the prior entanglement.
The prior entanglement can be replaced by user's upload of the entangled states.
However, we state our results with the prior entanglement because the entangled states can be prepared without any user's operation.

We prove the following two theorems.

\begin{theo}	\label{theo:31}
In the visible setting,
	there exists a QSPIR protocol for $\pureset{\mathbb{C}^2}$ 
	with $2\vF$-bit average classical communication,
	$8$-qubit average quantum communication,
	and
	$4$-ebit average prior entanglement.
\end{theo} 

\begin{theo}	\label{theo:32}
Let $d\geq 2$.
In the visible setting,
	there exists a QSPIR protocol for $\pureset{\mathbb{C}^d}$ 
	with 
	$2\vF$-bit average classical communication,
	$4d^d\log d$-qubit average communication,
	and $2d^d\log d$-ebit average prior entanglement.
\end{theo}

The protocols for pure states are constructed by the following idea.
Any pure state is written as $\sU|0\rangle$ with unitary matrices $\sU$ 
	which are decomposed by rotation operations and phase-shift operations.
With this fact, 
	the user requests the servers to apply certain rotation operations and phase-shift operations on bipartite entangled states so that the user can recover the targeted pure state after receiving the entangled state.
To guarantee the user secrecy, 
	we import to our protocols the same query structure of two-server classical PIR protocol by Chor et al. \cite{CGKS98}.

We denote by $\U{\mathbb{C}^d}$ the unitary group on $\mathbb{C}^d$
and define the set of pure states associated with $\U{\mathbb{C}^d}$ as 
	\begin{align}
	\mathcal{S}(\mathbb{C}^d) = \px*{ |\sU\rrangle \mid \sU \in \U{\mathbb{C}^d} }.
	\end{align}
For a subset $C$ of $\U{\mathbb{C}^d}$,
	we denote 
	\begin{align}
	\mathcal{S}(C) &\coloneqq \px*{ |\sU\rrangle \mid \sU \in C } \subset \pureset{\mathbb{C}^d\otimes \mathbb{C}^d},\\
	\mathcal{P}(C) &\coloneqq \px*{ \sU|0\rangle \mid \sU \in C } \subset \pureset{\mathbb{C}^d}.
	\end{align}
We can find the following relation between QSPIR protocols for $\mathcal{S}(C)$ and $\mathcal{P}(C)$.
\begin{prop}	\label{prop:chaa}
Let $C \subset \U{\mathbb{C}^d}$.
If there exists a QSPIR protocol for $\mathcal{S}(C)$ with communication complexity $c$,
	there exists a QSPIR protocol for $\mathcal{P}(C)$ which succeeds with probability $1/d$ 
	and 
	has the average communication complexity $cd$.
\end{prop}
\begin{proof}
We construct a QSPIR protocol for $\mathcal{P}(C)$ as follows.
After applying the QSPIR protocol for $\mathcal{S}(C)$ whose output state is $|\sU_k\rrangle$ on $A\otimes B$, 
	the user performs the basis measurement $\px*{ |0\rangle, \ldots, |d-1\rangle}$ on $A$.
If the measurement outcome is $0$, the resultant state on $B$ is $|\psi_k\rangle = \sU_k|0\rangle$.
Otherwise, repeat the process. 
Since the measurement outcome is $0$ with probability $1/d$,
	the expected number of trials is $d$.
Thus, the average communication complexity is $cd$.
\end{proof}

We also construct a QSPIR protocol for pure states represented by commutative unitary matrices.

\begin{theo} \label{theo:33}
Let $C \subset \U{\mathbb{C}^d}$ consist of commutative unitary matrices.
In the visible setting,
	there exists a QSPIR protocol for $\mathcal{S}(C)$ 
	with $2\vF$-bit classical communication,
	$4\log d$-qubit quantum communication,
	and
	$2\log d$-ebit prior entanglement.
\end{theo}
Combining Theorem~\ref{theo:33} with Proposition~\ref{prop:chaa}, we have the following corollary.
\begin{coro} \label{coro:33} 
Let $C \subset \U{\mathbb{C}^d}$ consist of commutative unitary matrices.
In the visible setting,
	there exists a QSPIR protocol for $\mathcal{P}(C)$ 
	with $2\vF$-bit classical communication,
	$8d\log d$-qubit average quantum communication,
	and
	$4\log d$-ebit average prior entanglement.
\end{coro}

In the following subsections, we construct the QSPIR protocols to achieve the communication complexity of Theorems~\ref{theo:31}, \ref{theo:32}, \ref{theo:33}. 
Since the protocol for Theorem~\ref{theo:33} is the most simplest and the similar idea is used in the other two protocols,
	we will first construct the protocol for Theorem~\ref{theo:33} and then the other two protocols.

\subsection{QSPIR protocol for pure states described by commutative unitaries}


In this subsection, we construct a two-server QSPIR protocol in the visible setting which achieves the communication complexity in Theorem~\ref{theo:33}.

\begin{prot}[Two-server QSPIR protocol for pure states described by commutative unitaries] \label{prot:wpe}
	For commutative $\vF$ unitaries $\sU_1, \ldots, \sU_\vF$ on $\mathbb{C}^d$, the message states are given as
	\begin{align}
	|\sU_1\rrangle, \ldots ,|\sU_\vF\rrangle	\in \mathbb{C}^{d}\otimes \mathbb{C}^{d}.	\label{eq:ins}
	\end{align}
	When the user's target index $K$ is $k\in[\vF]$, i.e., the targeted state is $|\sU_k\rrangle$, our protocol is given as follows.
\begin{enumerate}[leftmargin=1.5em]
\item \textbf{Query}: 
			The user chooses 
				$Q= (Q_{1},\ldots, Q_{\vF}) \in  \{0,1\}^{\vF}$ uniformly at random.
			The variable $Q' = (Q_{1}',\ldots,Q_{\vF}' )\in  \{0,1\}^{\vF}$ is defined as 
				\begin{align}
				Q_{\ell}' = 
					\begin{cases}  
					Q_{\ell} & \text{for $\ell\neq k$},	\\ 
					Q_{\ell}\oplus 1 & \text{for $\ell=k$}. 
					\end{cases}
				\end{align}
			The user sends $Q$ and $Q'$ to Server 1 and Server 2, respectively.

\item \textbf{Entanglement Sharing}: 
	Let $A, A', B, B'$ be qudits.
	Server 1 and Server 2 share 
				two maximally entangled state $|\sI_d\rrangle$ on $A\otimes A'$ and $B\otimes B'$,
				where Server 1 (Server 2) contains $A\otimes B$ ($A' \otimes B'$).

\item \textbf{Answer}: 
				{When $Q=q$ and $Q'=q'$,} we define
				\begin{align}
				(t_{1},\ldots, t_{\vF})   &= q\oplus (1,\ldots,1),\\
				(t_{1}',\ldots, t_{\vF}') &= q'\oplus (1,\ldots,1).
				\end{align}
				Server $1$ applies 
				\begin{align}
				\sU_1^{q_{1}} \cdots \sU_{\vF}^{q_{\vF}},\\
				\sU_1^{t_{1}} \cdots \sU_{\vF}^{t_{\vF}}
				\end{align}
				on $A$ and $B$, respectively,
				and sends $A$ and $B$ to the user.
				Similarly, Server $2$ applies 
				\begin{align}
				\bar{\sU}_1^{q_{1}'} \cdots \bar{\sU}_{\vF}^{q_{\vF}'},\\
				\bar{\sU}_1^{t_{1}'} \cdots \bar{\sU}_{\vF}^{t_{\vF}'}
				\end{align}
				on $A'$ and $B'$, respectively,
				and sends $A'$ and $B'$ to the user.
				

\item \textbf{Reconstruction}:
				The user outputs the state on $A_1\otimes A_2$ if $Q_{k} = 1$, otherwise outputs the state on $B_1\otimes B_2$. \QEDA
\end{enumerate}
\end{prot}

%

Protocol~\ref{prot:wpe} satisfies the correctness, secrecy, and communication complexity, desired in Theorem~\ref{theo:33}, which is shown as follows.

\begin{itemize}[leftmargin=1.5em]
\item
\textbf{Correctness:}
{Consider the case where $Q_{k} =1$.}
When $Q=q$, the state on $A\otimes A'$ after the measurement is 
	\begin{align}
	&(\sU_1^{q_{1}} \cdots \sU_{\vF}^{q_{\vF}} \otimes \bar{\sU}_1^{q_{1}'} \cdots \bar{\sU}_{\vF}^{q_{\vF}'}) |\sI_d \rrangle\\
	&=(\sU_1^{q_{1}} \cdots \sU_{\vF}^{q_{\vF}} \otimes \bar{\sU}_1^{q_{1}'} \cdots \bar{\sU}_{\vF}^{q_{\vF}'}) |\sI_d \rrangle\\
	&=|\sU_1^{q_{1}} \cdots \sU_{\vF}^{q_{\vF}} (\sU_{\vF}^{\dagger})^{q_{\vF}'} \cdots ({\sU}_1^{\dagger})^{q_{1}'} \rrangle\\
	&=|\sU_k \rrangle, \label{eq:commust} 
	\end{align}
where \eqref{eq:commust} follows from the commutativity of the unitaries $\sU_1,\ldots, \sU_{\vF}$,
	$q_{\ell}\oplus q_{\ell}' = \delta_{\ell,k}$, and $(q_{k},q_{k}') = (1,0)$.
By similar analysis, if $Q_{k}= 0$, the resultant state on $B\otimes B'$ is $|\sU_k\rrangle$. 

%

\item
\textbf{Secrecy:}
Throughout the protocol, the servers only obtain the queries, and 
	each query is uniformly random $\vF$ bits.
Therefore, each server does not obtain any information of $k$.
On the other hand, at the end of the protocol, the user obtains both of $|\sU_k\rrangle$ and $|\sU_k^{\dagger}\rrangle$.
Although $|\sU_k^{\dagger}\rrangle$ is transmitted additionally, no information of all other message states is leaked to the user.

\item
\textbf{Communication complexity:}
The query is $2\vF$ bits,
	the communication from the server to the user is $4$ qudits, i.e., $4\log d$ qubits,
	and prior entanglement is $2$ copies of $|I_d\rrangle$, i.e., $2\log d$ ebits.
	
\end{itemize}

\subsection{QSPIR protocol for pure qubit states}

In this subsection, we construct a two-server QSPIR protocol for pure qubit states in the visible setting which achieves the communication complexity in Theorem~\ref{theo:31}.

Define the rotation operation on $\mathbb{C}^2$ and the phase-shift operation by 
\begin{align*}
	\Rr(\theta) = 
	\begin{pmatrix}
	\cos\theta & - \sin\theta	\\
	\sin\theta & \cos\theta
	\end{pmatrix}
	,\quad
	\Ss(\varphi) = 
	\begin{pmatrix}
	e^{-\ii\varphi/2}	&	0	\\
	0	&	e^{\ii\varphi/2}
	\end{pmatrix}
\end{align*}
for $\theta, \varphi \in [0,2\pi)$.
For any $\varphi,\varphi', \theta,\theta'$,  we have 
\begin{align}
\Rr(\theta) \Rr(\theta') = \Rr(\theta+\theta'), \quad 
\Ss(\varphi) \Ss(\varphi') = \Ss(\varphi+\varphi'),
\end{align}
and therefore, 
	$\Ss(\varphi)$ and $\Ss(\varphi')$ ($\Rr(\theta)$ and $\Rr(\theta')$) are commutative.
We also have 
\begin{align}
\Rr(\theta)^{\top} = \Rr(-\theta)
,\quad 
\Ss(\varphi)^{\top} = \Ss(\varphi) 
\end{align}
and
\begin{align}
\sX\Rr(\theta)\sX = \Rr(-\theta),\quad 
\sX \Ss(\varphi) \sX = \Ss(-\varphi) 
.
\label{eq:xrxxsx}
\end{align}
As special cases, we have $\sY \coloneqq \sX\sZ = \Rr(\pi/2)$ and $\sZ = \Ss(\pi)$.
Any pure qubit states $|\psi\rangle$ are written with some $\varphi\in[0,2\pi)$ and $\theta \in [0,\pi/2]$ as 
\begin{align}
|\psi\rangle =  \Ss(\varphi)\Rr(\theta)|0\rangle.
\end{align}

Now, we construct a QSPIR protocol in the visible setting for qubit states, which achieves the communication complexity in Theorem~\ref{theo:31}.



\begin{prot}[Two-server QSPIR protocol for qubit states] \label{prot:npe}

For any message qubit states $|\psi_1\rangle,\ldots, |\psi_{\vF}\rangle \in \pureset{\mathbb{C}^2}$,
	we choose the parameters $\varphi_\ell$ and $\theta_\ell$ as
\begin{align}	
	|\psi_\ell\rangle = \Ss(\varphi_\ell)\Rr(\theta_\ell)|0\rangle.
\end{align}
	When the user's target index $K$ is $k\in[\vF]$, i.e., the targeted state is $|\psi_k\rangle$, our protocol is given as follows.
\begin{enumerate}[leftmargin=1.5em]
\item \textbf{Query}: 
			The same as Protocol~\ref{prot:wpe}.

\item \textbf{Entanglement Sharing}: 
	Let $A, A', B, B'$ be qubits.
	Server 1 and Server 2 share 
				two maximally entangled state $|\sI_2\rrangle$ on $A\otimes A'$ and $B\otimes B'$,
				where Server 1 contains $A\otimes B$ and
						Server 2 contains $A' \otimes B'$.

\item \textbf{Answer}: 
			{When $Q=q$,} 
				Server $1$ 
					calculates 
					$\tilde{\theta} \coloneqq  \sum_{\ell=1}^{\vF} q_{\ell} \theta_\ell$
					and 
					$\tilde{\varphi} \coloneqq  \sum_{\ell=1}^{\vF} q_{\ell} \varphi_\ell$, 
					applies $\Rr(-\tilde{\theta})$ and $\Ss(\tilde{\varphi})$  
					on $A$ and $B$, respectively,
					and sends $A\otimes B$ to the user.
				Similarly,
			{when $Q'=q'$,} 
					Server $2$ 
					calculates 
					$\tilde{\theta}' \coloneqq  \sum_{\ell=1}^{\vF} q_{\ell}' \theta_\ell$
					and
					$\tilde{\varphi}' \coloneqq  \sum_{\ell=1}^{\vF} q_{\ell}' \varphi_\ell$,
					applies $\Rr(-\tilde{\theta}')$ and $\bar{\Ss}(\tilde{\varphi}')$ 
					on $A'$ and $B'$, respectively,
					and sends $A'\otimes B'$ to the user.

\item \textbf{Reconstruction}:
	\begin{enumerate}
		\item
				The user performs the Bell measurement $\mathbf{M}_{\sX\sZ}$ defined in \eqref{mes}
					on $A'\otimes B'$.
				Depending on the measurement outcomes $(a,b) \in \{0,1\}^2$,
					the user applies $\sY^{-a} \otimes \sZ^{a+b} $ on $A\otimes B$
							and performs the basis measurement $\{|0\rangle, |1\rangle \}$ on $A$.
		\item
				If $Q_{k} =1$ and the measurement outcome is $0$, the state on $B$ is $|\psi_{k}\rangle$.
				If $Q_{k} =0$ and the measurement outcome is $1$, 
					the user applies $\sX$ operation on $B$ and then the state on $B$ is $|\psi_{k}\rangle$.
				This process succeeds with probability $1/2$.
				Otherwise, repeat from Step 2. 
				\QEDA
	\end{enumerate}
\end{enumerate}
\end{prot}

Protocol~\ref{prot:npe} satisfies the correctness, secrecy, and communication complexity, desired in Theorem~\ref{theo:31}, which is shown as follows.

\begin{itemize}[leftmargin=1.5em]
\item 
\textbf{Correctness:}
{When $Q=q$,} after the operations of the servers, 
	the states on $A\otimes A'$ and $B\otimes B'$ are
	\begin{align}
	\Rr(-\tilde{\theta}) \otimes \Rr(-\tilde{\theta}') |\sI\rrangle
	&= | \Rr(-\tilde{\theta})(\Rr(-\tilde{\theta}'))^\top \rrangle\\
	&= | \Rr(-\tilde{\theta}+\tilde{\theta}') \rrangle\\
	&= | \Rr((-1)^{q_{k}}\theta_k)\rrangle \in A\otimes A',
		\label{eq:Rrbi}
		\\
	\Ss(\tilde{\varphi}) \otimes \bar{S}(\tilde{\varphi}') |\sI\rrangle
	&= | \Ss(\tilde{\varphi})(\Ss(\tilde{\varphi}'))^\dagger \rrangle\\
	&= | \Ss(\tilde{\varphi}-\tilde{\varphi}') \rrangle\\
	&= | \Ss((-1)^{q_{k}+1}\varphi_k) \rrangle \in B\otimes B'.
		\label{eq:Ssbi}
	\end{align}
After the Bell measurement of the user with the measurement outcome $(a,b)$,
	the state on $A\otimes B$ is derived from \eqref{qe:feafterf} as 
	\begin{align}
	& | \Rr((-1)^{q_{k}}\theta_k) \sX^{a} \sZ^{-b} \Ss((-1)^{q_{k}+1}\varphi_k) ^{\top}\rrangle \\
	&= | \Rr((-1)^{q_{k}}\theta_k) \sY^{a} \sZ^{-a-b} \Ss((-1)^{q_{k}+1}\varphi_k) \rrangle \\
	&= | \sY^a  \Rr((-1)^{q_{k}}\theta_k) \Ss((-1)^{q_{k}+1}\varphi_k) \sZ^{-a-b} \rrangle.
	\label{eq:zkq1}
	\end{align}
After the user applies $\sY^{-a} \otimes \sZ^{a+b}$ on $A\otimes B$ in Step 4-a),
	the state \eqref{eq:zkq1} is changed as 
	\begin{align}
	| \Rr((-1)^{q_{k}+1}\theta_k) \Ss((-1)^{q_{k}+1}\varphi_k)  \rrangle
	\in A\otimes B,
	\label{eq:state_red}
	\end{align}
and if the basis measurement outcome in Step 4-b) is $x\in\{0,1\}$, the last state on $B$ is $(\Rr((-1)^{q_{k}}\theta_k) \Ss((-1)^{q_{k}+1}\varphi_k) )^{\top}|x\rangle
=\Ss((-1)^{q_{k}+1}\varphi_k) \Rr((-1)^{q_{k}+1}\theta_k) |x\rangle
$. 
If the measurement outcome $x$ is $0$ and $q_{k}=1$, the resultant state is $|\phi_k \rangle= \Ss(\varphi_k) \Rr(\theta_k)|0\rangle \in B$.
If the measurement outcome $x$ is $1$ and $q_{k}=0$, the resultant state after the user's operation $\sX$ is 
	\begin{align}
	\sX \Ss(-\varphi_k) \Rr(-\theta_k)|1\rangle 
	&= \sX \Ss(-\varphi_k)\sX \sX \Rr(-\theta_k)\sX|0\rangle \\
	&\stackrel{\mathclap{(e)}}{=}  \Ss(\varphi_k) \Rr(\theta_k)|0\rangle \\
	& = |\psi_k\rangle,
	\end{align}
where $(e)$ is from \eqref{eq:xrxxsx}.

For each execution of Step 4, 
	the user obtains $|\psi_k\rangle$ with probability $1/2$.
The probability to obtain $|\psi_k\rangle$ within $n$ repetition of Steps 3 and 4 is 
	$1-1/2^n$,
and the average number of repetitions is $2$. 


\item 
\textbf{Secrecy:}
	Each server does not obtain any information of $k$ since the query is the same as Protocol~\ref{prot:wpe}.
	At each repetition of Step 4, the user obtains the states of \eqref{eq:Rrbi} and \eqref{eq:Ssbi}, 
	which 
		only depends on the state $|\psi_k\rangle$.
	Thus, the user obtains no information of the other states.

\item 
\textbf{Communication complexity:}
The query is $2\vF$ bits.
For each execution of Steps 2-4, 
	necessary communication is $4$ qubits
	and 
	necessary prior entanglement is $2$ ebits.
Since the query is sent only once at Step 1, the classical communication complexity is $2\vF$ bits.
Thus, the average quantum communication complexity is $8$ qubits,
	and the average size of entanglement is $4$ ebits.
If we replace the prior entanglement by the transmission of $|\sI_2\rrangle\otimes |\sI_2\rrangle$ from the user to the servers,
	the average communication complexity is $2\vF$ bits and $16$ qubits.

\end{itemize}

Compared to the QPIR's trivial solution of downloading all quantum messages, which requires the communication of $\vF$ qubits,
	the above protocol has less quantum communication on average when $\vF > 16$ even if ebits are uploaded by the user.
On the other hand, 
	our protocol 
	requires classical communication of $2\vF$ bits.



\subsection{QSPIR protocol for pure qudit states ($d\geq 2$)}

In this subsection, we construct a two-server QSPIR protocol for pure qudit states in the visible setting which achieves the communication complexity in Theorem~\ref{theo:32}.

Similar to Protocol~\ref{prot:npe},
	we first consider the parameterization of pure states on $d$-dimensional systems.
{Define $d \times d$ matrices $R(\theta^1,..,\theta^{d-1} )$ and $S(\varphi^1,..,\varphi^{d-1} )$ as}
\begin{align}
\Rr(\theta^1,\ldots, \theta^{d-1})
 &= \Rr_{1}(\theta^1) \cdots \Rr_{d-1}(\theta^{d-1}),\\
\Ss(\varphi^1,\ldots, \varphi^{d-1})
 &= |0\rangle\langle 0| + \sum_{s=1}^{d-1} e^{\ii \varphi^s} |s\rangle\langle s|\\
 &= \begin{pmatrix}
 1 & 0 & 0 & 0 \\
 0 & e^{\ii \varphi^1} & 0 & 0 \\
 0 & 0 & \ddots & 0 \\
 0 & 0& 0& e^{\ii \varphi^{d-1}}
 \end{pmatrix},
\end{align}
where $\Rr_{s}(\theta)$ is the rotation 
	\begin{align}
	\begin{pmatrix}
	\cos\theta & - \sin\theta	\\
	\sin\theta & \cos\theta
	\end{pmatrix}
	\end{align}
	with respect to the two basis elements $|s-1\rangle$ and $|s\rangle$.
Notice that 
	$\Ss(\varphi^1,\ldots, \varphi^{d-1} )$ for all $\varphi^1,\ldots,\varphi^{d-1}$ ($\Rr_s(\theta^s)$ for all $\theta^s$) are commutative
	but any two of $\Rr_{1}(\theta^1), \ldots, \Rr_{d-1}(\theta^{d-1})$ are not in general.
We also have
$\Ss(\varphi^1,\ldots, \varphi^{d-1} )^\top = \Ss(\varphi^1,\ldots, \varphi^{d-1} )$,
	$\Rr_{s}(\theta^s)^{\top} = \Rr_{s}(-\theta^s)$,
	and
$$\sZ_{d} = \Ss\paren*{ \frac{2\pi \ii}{d}, \frac{4\pi \ii}{d},\ldots , \frac{2(d-1)\pi  \ii}{d} }.$$
It can be easily checked that any pure state $|\psi\rangle\in\pureset{\mathbb{C}^d}$ is written 
	in the form
\begin{align}
|\psi\rangle = 
		\Ss(\varphi^1,\ldots, \varphi^{d-1})
		\Rr(\theta^1,\ldots, \theta^{d-1})
		|0\rangle
\end{align}
	with $\varphi^1, \ldots,\varphi^{d-1}\in [0,2\pi)$ and $\theta^1,\ldots,\theta^{d-1} \in [0,\pi/2]$.

\begin{prot}[Two-server QSPIR protocol for qudit states] \label{prot:dlev}
For any message pure states $|\psi_1\rangle,\ldots, |\psi_{\vF}\rangle \in \pureset{\mathbb{C}^d}$,
	we choose the parameters $\varphi_\ell^1,\ldots, \varphi_\ell^{d-1}$ and $\theta_\ell^1,\ldots, \theta_\ell^{d-1}$ as
\begin{align}	
	|\psi_\ell\rangle = 
		\Ss(\varphi^1_{\ell},\ldots, \varphi^{d-1}_{\ell})
		\Rr(\theta^1_{\ell},\ldots, \theta^{d-1}_{\ell})
		|0\rangle.
\end{align}
When the user's target index $K$ is $k\in[\vF]$, i.e., the targeted state is $|\psi_k\rangle$, our protocol is given as follows.

\begin{enumerate}[leftmargin=1.5em]
\item \textbf{Query}: 
			The same as Protocol~\ref{prot:wpe}.

\item \textbf{Entanglement Sharing}: Server 1 and Server 2 share $d$ copies of the maximally entangled state $|\sI_d\rrangle$.

\item \textbf{Answer}:
	Similar to Step 3 of Protocol~\ref{prot:npe} and as analyzed in \eqref{eq:Rrbi} and \eqref{eq:Ssbi}, 
	%
	%
	Server 1 and Server 2 can jointly generate the following quantum states
		\begin{align}
		|\Ss(\bm{\varphi}_k) \rrangle,
		|\Rr_1(\theta_k^1) \rrangle,
		\ldots,
		|\Rr_{d-1}(\theta_k^{d-1}) \rrangle,
		\label{eq:inversesta0}
		\end{align}
		and 
		\begin{align}
		|\Ss(-\bm{\varphi}_k) \rrangle,
		|\Rr_1(-\theta_k^1) \rrangle,
		\ldots,
		|\Rr_{d-1}(-\theta_k^{d-1}) \rrangle.
		\label{eq:inversesta}
		\end{align}
		The servers generate and send these states to the user if it is requested in Step 4.

\item \textbf{Reconstruction}: 
	For this step, remind \eqref{qe:feafterf}:
	If the state is $|\sA\rrangle\otimes|\sB\rrangle$ and Bell measurement $\mathbf{M}_{\sX\sZ}$ defined in \eqref{mes} is performed with outcome $(a,b)$,
	the resultant state is
	$|\sA\sX^a\sZ^{-b} \sB^\top\rrangle$.
	\begin{enumerate}[leftmargin=0.3em]
%
%

	\item 
	The user requests the servers to send
		$|\Ss(\bm{\varphi}_k) \rrangle \in A\otimes A'$ and 
		$|\Rr_1(-\theta_k^1) \rrangle \in B\otimes B'$.
		The user performs the Bell measurement $\mathbf{M}_{\sX\sZ}$ defined in \eqref{mes} on $A'\otimes B'$. 
		If the measurement outcome is $(0,\alpha_1)$ for some $\alpha_1$, which happens with probability $1/d$,
				the user obtains the outcome of the following conversions 
			\begin{align}
			|\Ss(\bm{\varphi}_k) \rrangle |\Rr_1(-\theta_k^1) \rrangle 
			&\mapsto 
			|\Ss(\bm{\varphi}_k) \sZ^{\alpha_1} \Rr_1(\theta_k^1)\rrangle\\
			&=
			|\sZ^{\alpha_1} \Ss(\bm{\varphi}_k) \Rr_1(\theta_k^1)\rrangle,
			\end{align}
			where $\mapsto$ represents the state reduction by the measurement.
		Otherwise,
				the user repeats Step 4-a.

	\item 
			Next, the user 
				requests
				$|\Rr_2(-\theta_k^2) \rrangle$ to the servers and 
				performs 
				the same measurement $\mathbf{M}_{\sX\sZ}$. 
				If the measurement outcome is $(0,\alpha_2)$ for some $\alpha_2$, the user obtains the outcome of the following conversion:
			\begin{align*}
			&|\sZ^{\alpha_1} \Ss(\bm{\varphi}_k) \Rr_1(\theta_k^1) \rrangle |\Rr_2(-\theta_k^2) \rrangle \\
			&\mapsto |\sZ^{\alpha_1+\alpha_2} \Ss(\bm{\varphi}_k)\Rr_1(\theta_k^1)\Rr_2(\theta_k^2)\rrangle,
			\end{align*}
			If it fails to measure $(0,\alpha_2)$, restarts from {Step~{4-a}}.
			Repeating the similar conversions with $|\Rr_3(-\theta_k^3)\rangle$, \ldots, $|\Rr_{d-1}(-\theta_k^{d-1})\rangle$,
			the user obtains
			\begin{align}
			|\sZ^{\sum_{s=1}^{d-1} \alpha_s} \Ss(\bm{\varphi}_k)\Rr_1(\theta_k^1)\cdots \Rr_{d-1}(\theta_k^{d-1})\rrangle.
			\label{eq:wow}
			\end{align}
			
	\item The user performs Steps 4-a, 4-b with the states of \eqref{eq:inversesta} instead of \eqref{eq:inversesta0}.
		The user finally obtains $\beta_1,\cdots,\beta_{d-1}\in [0:d-1]$ and 
		\begin{align}
			|\sZ^{\sum_{s=1}^{d-1}\beta_s} \Ss(-\bm{\varphi}_k)\Rr_1(-\theta_k^1)\cdots \Rr_{d-1}(-\theta_k^{d-1})\rrangle.
			\label{eq:wow2}	
		\end{align}

	\item 
	Let $X\otimes X'$ be the system of two qudits on which
	the state \eqref{eq:wow} 
		is.
	The user performs the basis measurement $\{|0\rangle,\ldots, |d-1\rangle\}$ on $X'$.
	If the measurement outcome is $0$, 
		the user applies $\sZ^{-{\sum_{s=1}^{d-1}\alpha_s}}$ on $X$
		and then
		the resultant state on $X$ is the targeted state $|\psi_k\rangle$.
	This step is written as 
	\begin{align}
	&|\sZ^{{\sum_{s=1}^{d-1}\alpha_s}} \Ss(\bm{\varphi}_k) \Rr_1(\theta_k^1)\cdots \Rr_{d-1}(\theta_k^{d-1})\rrangle\\
	&\mapsto
	\sZ^{{\sum_{s=1}^{d-1}\alpha_s}} \Ss(\bm{\varphi}_k) \Rr_1(\theta_k^1)\cdots \Rr_{d-1}(\theta_k^{d-1})|0 \rangle \\
	&\mapsto 
	\Ss(\bm{\varphi}_k) \Rr_1(\theta_k^1)\cdots \Rr_{d-1}(\theta_k^{d-1})|0 \rangle 
	\\
	&= |\psi_k\rangle .
	\end{align}
	Otherwise, repeat from Step~4-a. 
	\QEDA
	\end{enumerate}

\end{enumerate}
\end{prot}


Protocol~\ref{prot:dlev} satisfies the correctness, secrecy, and communication complexity, desired in Theorem~\ref{theo:32}, which is shown as follows.

\begin{itemize}[leftmargin=1.5em]
\item
\textbf{Correctness:}
As described in Step 4, 
	the user recovers $|\psi_k\rangle$ with positive probability by repeating until success.

\item
\textbf{Secrecy:}
This protocol satisfies the user secrecy, because the queries are the same as Protocol~\ref{prot:npe} and the number of requests in Step 4.
Step 4-d is necessary since the states in \eqref{eq:inversesta0} and 
		\eqref{eq:inversesta} should be requested the same number of times.
This protocol also satisfies the server secrecy, because 
	the only information that the user obtains from the servers is states in \eqref{eq:inversesta0} and \eqref{eq:inversesta}
	and 
		these states only depend on the state $|\psi_k\rangle$.

\item
\textbf{Communication complexity:}
At Step 1, the size of queries is $2\vF$ bits.
For downloading all states \eqref{eq:inversesta0}, $2d\log d$-qubit communication and $d\log d$-ebit are required. 
The probability to succeed to generate the state in \eqref{eq:wow} is $1/d^{d-1}$
			and therefore, 
	the expected number of execution of Steps from 4-a to 4-c is $d^{d-1}$.
Thus, for finishing Steps from 4-a to 4-c, 
	the average quantum communication is less than $2d^{d}\log d = d^{d-1} \cdot 2d\log d$ qubits
	and 
	the average shared entanglement is less than $d^{d}\log d = d^{d-1} \cdot d\log d$ ebits.
Similarly, for Step 4-d, 
	the same average quantum communication and the same average shared entanglement are required.
	Thus, 
	Protocol~\ref{prot:dlev} requires $2\vF$-bit classical communication,
	$4d^d\log d$-qubit average quantum communication,
	and $2d^d\log d$-ebit average prior entanglement.
\end{itemize}

The communication complexity of Protocol~\ref{prot:dlev} does not scale with the number of message states $\vF$.
Since the QPIR's trivial solution of downloading all messages requires $\vF\log d$-qubit quantum communication,
	Protocol~\ref{prot:dlev} is more efficient on average if the number of messages $\vF$ is greater than $4d^d$.

\begin{remark}[Comparison of Protocols~\ref{prot:npe} and \ref{prot:dlev} for qubit states]
Both Protocols~\ref{prot:npe} and \ref{prot:dlev} can be applied for QPIR of qubit states.
However, 
	the average communication complexity of Protocol~\ref{prot:npe}, which is $8$ qubits and $4$ ebits,
	is less than that of 
		Protocol~\ref{prot:dlev} for $d=2$, which is $16$ qubits and $8$ edits.
The advantage of Protocol~\ref{prot:npe} comes from 
	the commutation relation between the rotation operations $\Rr(\theta)$ and $\sY = \sX\sZ = \Rr(\pi/2)$.
Since the reduced state after the Bell measurement is represented with $\sY,\sZ$ as \eqref{eq:zkq1}
	and $\sY$ is commutative with $\Rr(\theta)$ for any $\theta$,
	the user can reconstruct the targeted state efficiently as \eqref{eq:state_red}.
On the other hand, 
	the similar method cannot be applied for Protocol~\ref{prot:dlev}
	because the rotation operations $\Rr(\bm{\theta})$ for $d > 2$ are not commutative with any multiplication of $\sX_d$ and $\sZ_d$.
Thus, without using the commutative relation, 
	Protocol~\ref{prot:dlev} for $d=2$ results in greater communication complexity.
\end{remark}

\section{Conclusion} \label{sec:conclusion}

We have studied quantum private information retrieval for quantum messages in the blind and visible settings.
We have proved that the trivial solution of downloading all messages is optimal for honest one-server QPIR,
	which is a similar result to the classical PIR but different from QPIR for classical messages.
In the one-round case,
	the optimality is proved for both blind and visible settings, 
	and in the multi-round case, it is proved for the blind setting.
On the other hand,
	we have constructed efficient QPIR protocols with two cases.
The first case is one-server QPIR with prior entanglement in the blind setting.
We have constructed a reduction from any QPIR protocol for classical messages to a QPIR protocol for quantum messages
	and derived an efficient QPIR protocol from the protocol by Kerenidis et al. \cite{KLLGR16}.
The second case is multi-server QSPIR in the visible setting.
We have constructed three protocols for pure qubit states, pure qudit states, and pure states described by commutative unitaries.
These protocols in the visible setting are symmetric QPIR protocols in which the user obtains no information other than the targeted message state. 
Thus, we proved that QSPIR is possible for quantum messages.
Furthermore, these protocols are also more efficient than the QPIR's trivial solution when the number of messages are sufficiently greater than the dimension of quantum systems.

There are many open problems related to the study of QPIR for quantum messages.
Our optimality proof of the trivial one-server QPIR protocol is 
	for the one-round visible setting and the multi-round blind setting. 
It is unknown whether the optimality still applies 
	for 
	the multi-round one-server visible setting
	and the multi-server blind setting.
Also, we have only considered QPIR with information-theoretic security.
Our result does not preclude the possibility of the non-trivial one-server QPIR protocol with computational security.
Furthermore, 
	even if our QSPIR protocol for qudits is more efficient than the QPIR's trivial solution when the number of message states $\vF$ is sufficiently large,
		the average communication complexity increases exponentially with the dimension $d$ of the system. 
In addition, our QSPIR protocols are for pure states and it cannot be directly extended to protocols for the mixed state.
Thus, more efficient QSPIR protocols for qudits and QSPIR for mixed states are also open problems.
Interesting applications of our QSPIR protocols can also be considered for other communication and computation problems.
We leave these questions for interested readers.


%

	

\section{Acknowledgement}

The authors are grateful to Prof. Fran\c{c}ois Le Gall for helpful discussion and comments.
SS is supported by JSPS Grant-in-Aid for JSPS Fellows No.\ JP20J11484.
MH was supported in part by 
Guangdong Provincial Key Laboratory (Grant, No. 2019B121203002).

%
%
%
%
%
%
%
%
%

\end{document}